\documentclass[letterpaper,11pt]{article}
\usepackage{verbatim}
\usepackage{xspace}
\usepackage{amsmath}
\usepackage{amsthm}
\usepackage{amssymb}
\usepackage{graphicx}
\usepackage{epstopdf}
\usepackage{url}
\usepackage{fontenc}
\usepackage[utf8]{inputenc}
\usepackage{mathtools}

\usepackage{fullpage}
\usepackage{verbatim}
\usepackage{paralist}
\usepackage{amsfonts}
\usepackage{amsmath}
\usepackage{amsthm}
\usepackage{latexsym}
\usepackage{url}
\usepackage{algorithm,algorithmic}

\usepackage{thmtools}
\usepackage{thm-restate}
\usepackage{hyperref}
\usepackage{cleveref}
\usepackage{complexity}
\usepackage{color}

\newcommand{\INPUT}{\item[{\bf Input:}]}

\newcommand{\RR}{{\mathbb R}}
\newcommand{\myvec}[1]{\mathbf{#1}}

\newcommand{\ignore}[1]{}%

\newif\ifFULL
\FULLtrue

\title{Locally Adaptive Optimization: Adaptive Seeding for Monotone Submodular Functions}

\author{
Ashwinkumar Badanidiyuru\\
 Google\\
  \texttt{ashwinkumarbv@gmail.com}\\
   \and
  Christos Papadimitriou\\
 UC Berkeley\\
  \texttt{christos@cs.berkeley.edu}\\
  \and
  Aviad Rubinstein\\
 UC Berkeley\\
  \texttt{aviad@cs.berkeley.edu}\\
\and
  Lior Seeman\\
 Cornell Univeristy\\
  \texttt{lseeman@cs.cornell.edu}\\
  \and
  Yaron Singer\\
  Harvard University\\
  \texttt{yaron@seas.harvard.edu}
}

\newtheorem{theorem}{Theorem}[section]
\newtheorem{cor}[theorem]{Corollary}
\newtheorem{lemma}[theorem]{Lemma}

\newtheorem{definition}[theorem]{Definition}

\newtheorem*{theorem*}{Theorem}
\newtheorem*{corollary*}{Corollary}
\newtheorem*{conjecture*}{Conjecture}
\newtheorem*{lemma*}{Lemma}
\newtheorem*{thm*}{Theorem}
\newtheorem*{prop*}{Proposition}
\newtheorem*{obs*}{Observation}
\newtheorem*{rem*}{Remark}
\newtheorem*{definition*}{Definition}

\begin{document}
\frenchspacing
\sloppy
\widowpenalty10000
\clubpenalty10000
\date{}
\maketitle

\begin{abstract}
The Adaptive Seeding problem is an algorithmic challenge motivated by influence maximization in social networks: One seeks to select among certain accessible nodes in a network, and then select, adaptively, among neighbors of those nodes as they become accessible in order to maximize a global objective function.  More generally, adaptive seeding is a stochastic optimization framework where the choices in the first stage affect the realizations in the second stage, over which we aim to optimize.  

Our main result is a $(1-1/e)^2$-approximation for the adaptive seeding problem for any monotone submodular function. While adaptive policies are often approximated via non-adaptive policies, our algorithm is based on a novel method we call \emph{locally-adaptive} policies. These policies combine a non-adaptive global structure, with local adaptive optimizations. This method enables the $(1-1/e)^2$-approximation for general monotone submodular functions and circumvents some of the impossibilities associated with non-adaptive policies.

We also introduce a fundamental problem in submodular optimization that may be of independent interest: given a ground set of elements where every element appears with some small probability, find a set of expected size at most $k$ that has the highest expected value over the realization of the elements. We show a surprising result: there are classes of monotone submodular functions (including coverage) that can be approximated almost optimally as the probability vanishes. For general monotone submodular functions we show via a reduction from \textsc{Planted-Clique} that approximations for this problem are not likely to be obtainable. This optimization problem is an important tool for adaptive seeding via non-adaptive policies, and its hardness motivates the introduction of \emph{locally-adaptive} policies we use in the main result.  
\end{abstract}
\newpage
\pagenumbering{arabic}
\section{Introduction}
The surge of massive digital records of human interactions in recent years provides a new system-wide perspective on social networks.  In addition to observing and predicting patterns of collective human behavior, in many cases the dynamics of the network can be engineered.
One such example is when attempting to initiate a large cascade by seeding it at certain important nodes in the network to promote a product or social movement through word-of-mouth.
The algorithmic challenge of selecting individuals who can serve as early adopters of a new idea, product, or technology in a manner that will trigger a large cascade in the social network is known as \emph{influence maximization}.  Since it was first posed by Domingos and Richardson~\cite{DR01,RD02} and elegantly formulated and further developed by Kempe, Kleinberg, and T{\'a}rdos~\cite{KKT03}, a broad range of algorithmic methods have been developed for this canonical problem~\cite{BHMW11,borgs2012influence,C08,RLK10,LAH06,LKGFVG07,KDD11,MR07}.

In many applications of influence maximization, despite having full knowledge of the network, one may only have access to a small slice of the network. In marketing applications for example, companies often reward influential users who visit their online store, or who have engaged with them in other ways (subscribe to a mailing list, follow the brand, install an application etc.).  If we think of users who arrive at a store or follow a brand as being randomly sampled from the network, observing high-degree users is a rare event.  This is simply due to the heavy-tailed degree distributions of social networks.  Since influence maximization techniques are based on selecting high degree users (not necessarily the \emph{highest} degree), their application on such samples can become ineffective.
In general, access to high degree individuals in social networks is often rare, which raises the following question.  

\begin{center}
\emph{Is it possible to design effective influence maximization strategies despite the rarity of influencers?}
\end{center}
 
To tackle the problem or rare influencers, the adaptive seeding framework was recently developed in~\cite{SS13}.  The framework is a stochastic optimization model which formalizes an intuitive approach: rather than spend the entire budget on the non-influential users, we can spend a fraction of the budget on the accessible users, wait for their friends to appear as a result, and optimize influence by using the remaining budget to seed influential friends.  
The idea is to leverage what's known as the \emph{friendship paradox}~\cite{feld1991}, which suggests that although people are not likely to be influential, they are likely to know someone who is.  In~\cite{LS15} it was shown that in well-established mathematical models of social networks there are asymptotic gaps between the degree of a random node and its neighbor.  This structural property implies that dramatic improvements to influence maximization are indeed achievable by optimizing influence over friends.  In recent work~\cite{HS15}, along with scalable algorithms for this problem, it was shown through various experiments on the Facebook graph that dramatic improvements to naive application of influence maximization are indeed obtainable through adaptive seeding.

\paragraph{The adaptive seeding model.}
The adaptive seeding model is a \emph{two-stage} stochastic optimization framework. 
We are given a set of nodes $X$, their set of neighbors $\mathcal{N}(X)$, each associated with a probability $p_i$, as well as a budget $k\in \mathbb N$ and a function $f:2^{\mathcal{N}(X)} \to \mathbb R$.  In the first stage, a set $S \subseteq X$ can be selected, which causes each one of its neighbors to materialize independently with probability $p_i$.  In the second stage, the remainder of the budget can be used to optimize the function $f(\cdot)$ over the realized neighbors.  This function quantifies the expected number of individuals in the network that will be influenced as a result of selecting a subset of early adopters.  The goal is to select a subset $S \subseteq X$ of size at most $k$ s.t. the function can be optimized in expectation over all possible realizations of its neighbors with the remaining budget $k-|S|$.  
Equivalently, our goal is to select $S \subseteq X$ s.t. in expectation over all possible realizations $R_{1},\ldots,R_{m}$ of $\mathcal{N}(X)$ the value of a set of its neighbors $T_i$ of size $k-|S|$ that appears in the realization $R_i$ is optimal.  

Stated in these terms the objective is:
\begin{align*}
\max_{S \subseteq X} \sum_{i=1}^{m}f(T_i)\cdot p(R_i) & \hspace{0.2in}\\
T_i \subseteq R_i \cap \mathcal{N}(S) &\hspace{0.2in} \forall i \in [m] \\
|S|+ |T_i| \leq k & \hspace{0.2in} \forall i \in [m]
\end{align*}

The adaptive seeding formulation models the challenge of recruiting neighbors who can become effective influencers using the budget, rather than trying to influence them to forward the information without incentives as in the standard influence maximization framework.%
\footnote{Having probabilities on the nodes in $\mathcal{N}(X)$ and not the edges implies that the neighbors are not more likely to appear if they have more parents in $X$ selected by the algorithm.  This corresponds to what's known in microeconomics as a standard Bayesian utility model with no externalities.  For submodular functions, this problem can be shown to be equivalent to the one where probabilities appear on the edges and not the nodes.  This corresponds to a model where nodes in the second stage are influenced by nodes in the first stage through the independent cascade model~\cite{KKT03}.}
Although it seems quite plausible that the probabilities of attracting neighbors could depend on the rewards they receive, the simplification assuming unit costs is deliberate. 
This simplification is consistent with the celebrated Kempe-Kleinberg-T{\'a}rdos model~\cite{KKT03}, and can be extended to the case where nodes take on different costs~\cite{RSS14}.  

While the motivation is influence maximization, adaptive seeding is a versatile framework of stochastic optimization
\ifFULL
(see related work section for further discussion on stochastic optimization)%
\fi
. At its essence, it formalizes the following question. 

\begin{center}
\emph{Given a distribution on the consequences of the actions we take in the present, can we design algorithms that optimize events in the future? 
}
\end{center}

The main question when considering stochastic optimization models is whether the same guarantees as standard optimization can be obtained~\cite{shmoys2004stochastic}.  It is not hard to show that for very simple objective functions such as $f(S)=|S|$, the two-stage optimization is NP-hard even in the non-stochastic case, i.e. when all probabilities are one.  Thus, approximation is needed.  Therefore, in the context of adaptive seeding the question is whether an objective that can be well approximated in standard optimization can also be well approximated in adaptive seeding.

\paragraph{Submodularity.}  A function $f:2^{N} \to \mathbb{R}_{+}$ is submodular if $f(S \cup T) \leq f(S) + f(T) - f(S \cap T)$.  Equivalently, a function is submodular if it has a natural diminishing returns property: for any $S\subseteq T \subseteq N$ and $a \in N\setminus T$ a function is submodular if $f_{S}(a) \geq f_{T}(a)$, where $f_{A}(B) = f(A \cup B) - f(A)$ for any $A,B \subseteq N$.  Unless otherwise stated, we will assume the function is normalized ($f(\emptyset)=0$) and \emph{monotone} ($S\subseteq T$ implies $f(S) \leq f(T)$).  In standard optimization submodularity is a certificate for desirable approximation guarantees (1-1/e approximation for maximizing such functions under a matroid constraint~\cite{V08}), and a slightly broader class known as \emph{fractionally subadditive} functions already exhibits strong information-theoretic lower bounds~\cite{MSV08}.

\paragraph{Adaptive Seeding of submodular functions.}  The effectiveness of adaptive seeding depends crucially on our ability to optimize general classes of submodular functions.  Say that a function can be \emph{adaptively seeded} if the adaptive seeding problem can be approximated within a constant factor for this objective.  The main result in~\cite{SS13} shows that any function in a class known as the \emph{triggering model} --- a special class of monotone \emph{submodular} functions defined in the seminal work of Kempe, Kleinberg, and T{\'a}rdos~\cite{KKT03} --- can be adaptively seeded.  
But far more general models of submodular functions are used to describe influence in social networks~\cite{KKT03,MR07}, and the techniques of~\cite{SS13} cannot be applied to these
\ifFULL
(see related work section
\fi
for further discussion on submodular function, their connection with the influence maximization problem and the limitation of previous techniques for adaptive seeding). Thus the main question in this work is:
\begin{center}
\textit{Can any submodular function be adaptively seeded?}
\end{center} 

Naturally, we will seek algorithms that obtain the best approximation ratio possible.  The main challenge would be to obtain an algorithm that achieves a $(1-1/e)^2$ approximation ratio for general monotone submodular functions.  This bound is a natural goal for this problem, as we discuss below.

\subsection{Warm up: the non-stochastic case}
We begin by considering the non-stochastic version of the problem.  That is, the version in which every node in the set of neighbors appears with probability one.  Here we are given a set $X$ and its neighbors $\mathcal{N}(X)$, there is a monotone submodular function defined on $\mathcal{N}(X)$, and the goal is to select $t\leq k$ elements in $X$ connected to a set of size at most $k-t$ in $\mathcal{N}(X)$ for which the submodular function has the largest value.  
A trivial solution would be to run a greedy algorithm with budget of $k/2$ on parent-child pairs.  That is, at every stage select the node in $\mathcal{N}(X)$ whose marginal contribution given the nodes already selected is the largest, and add one of its parents (if none have been added in previous rounds) to ensure the solution is feasible.  Since we need at most $k/2$ parents to select $k/2$ children the solution is feasible, and submodularity guarantees the solution's value is at least half of the value as if we were to run the algorithm with a budget of $k$, which is an upper bound on the optimal solution.  It is well known that the greedy algorithm is a $1-1/e$-approximation to the optimal solution, and hence this trivial algorithm would be a $(1-1/e)/2$-approximation.  

%
%
%
%
\paragraph{Optimal approximation via $\epsilon$-blocks.}  
A natural extension to the above approach would be to pair parent nodes (i.e. nodes in $X$) with \emph{subsets} of their children.  We call such pairs \emph{blocks}, and the \emph{density} of a block is simply the ratio between its marginal contribution and its size: for a set of children $T \subseteq \mathcal{N}(X)$, the marginal density of a block $(x,B)$ with respect to $T$ is
$$f_{T}(B)/(1+|B|).$$
Ideally, for each parent we would add the subset of children which makes for the densest block, as one can then show that an algorithm which iteratively adds the densest block results in the optimal $1-1/e$ approximation.  However, even for coverage functions finding the densest block implies solving an NP-hard problem.  Instead of densest blocks we can consider using $\epsilon$-blocks: a node in $x \in X$ and a subset of its neighbors of size at most $1/\epsilon$.  Note that for any constant $\epsilon>0$, finding the densest $\epsilon$-block can be done in polynomial time by brute forcing all subsets of $1/\epsilon$ neighbors of every node $x \in X$ and computing their value.  Importantly, one can show that for any $\epsilon>0$ the densest $\epsilon$-block $B$ is a $(1-\epsilon)$ approximation to the densest block $O$ if $|O|$ is larger than $1/\epsilon$:
%
$$ \frac{f_T(B)}{1+|B|} \geq   \left (\frac{1 }{1+1/\epsilon} \right )\frac{1/\epsilon}{|O|} f_T(O) = \left (\frac{1}{1+\epsilon}\right) \frac{f_T(O)}{|O|} \geq \Big ( 1-\epsilon \Big )\frac{f_T(O)}{1+|O|}$$
where the first inequality is by submodularity and taking the $1/\epsilon$ elements of $O$ with highest value.

The algorithm is now simple: until exhausting the budget, add the densest $\epsilon$-block to the solution.  For any $\epsilon>0$, this algorithm is a $(1-1/e -\epsilon$) approximation, and the idea of the analysis is natural; at every stage this algorithm selects an $\epsilon$-block that is a $(1-\epsilon)$ approximation to the densest block, and by applying a standard inductive argument one can show that this results in a $(1-1/e-\epsilon)$ approximation.
Although it only holds in the dummy non-stochastic version of the problem, this approach encapsulates the core idea in this paper.

\subsection{Synopsis}
In the stochastic case the optimal solution is an \emph{adaptive} policy: one which selects a subset from  $X$, and \emph{after} the realization of its neighbors, selects an optimal solution with its remaining budget.  Since adaptive policies are notorious in stochastic optimization for their difficulty, the standard approach is to design non-adaptive policies which approximate adaptive policies well.  Informally, a \emph{non-adaptive} policy is a policy which selects the subset $S$ and a set of its neighbors, \emph{a priori} to their realization. We cannot hope to obtain an approximation better than $1-1/e$ for the optimal non-adaptive policy unless P=NP~\cite{F98}, as the non-stochastic case is a special case. As we show later, the ratio between non-adaptive policies and adaptive ones can be as bad $1-1/e$, and we therefore naturally seek algorithms whose approximation ratio is $(1-1/e)^2$.
We do not know whether $(1-1/e)^2$ is the optimal approximation ratio, and it is one of the main open questions in this paper.

Similar to the exposition in the warmup above, finding a $(1-1/e)$ approximation to the optimal non-adaptive policy (and hence a $(1-1/e)^2$ approximation to the optimal adaptive policy) reduces to finding $\epsilon$-blocks with arbitrarily good precision in \emph{the stochastic case}.
The key challenge in computing $\epsilon$-blocks here reduces to the following fundamental problem (which may be of independent interest) we call \emph{submodular-optimization-with-small-probabilities (SOSP)}: given a ground set of elements where every element appears with some small probability, find a set of \emph{expected} size at most $k$ that has the highest expected value over the realization of the elements. 
While for some classes of submodular functions near-optimal solutions for this problem can be obtained, arbitrarily good approximations for this problem are not likely to be obtainable in general. 
In other words, non-adaptive policies do not suffice for getting  a $(1-1/e)^2$-approximation algorithm. 

Our main result builds on an alternative strategy for defining $\epsilon$-blocks.
Instead of non-adaptive policies we employ what we call {\em locally-adaptive policies}. Intuitively, a locally adaptive policy consists of a set of $\epsilon$-blocks, where within each block the policy can make \emph{adaptive decisions}.  
The adaptivity within a block lets us find the optimal $\epsilon$-block and enables the $(1-1/e)^2$ guarantee.

%

\subsection{Results}
\begin{itemize}

\item \textbf{Tight Adaptivity gap.} In Section~\ref{sec:RR-adaptive} we show that non-adaptive policies approximate adaptive policies within almost a factor of $1-1/e$.  We then show that this gap is tight.  

\item \textbf{Algorithm for SOSP.} In section~\ref{apx-rr} we introduce this problem and show how it can be solved almost optimally for matroid-rank-sum functions (which include coverage functions) by convex programming as the probabilities vanish. We then show how to use this problem to design a $(1-1/e)^2$-approximation algorithm for this class of submodular functions. 

\item \textbf{Hardness of SOSP.} We show that for general monotone submodular functions, the problem is hard to approximate to arbitrary precision by a reduction from the \textsc{Planted-Clique} problem. Thus, for general submodular functions computing the optimal $\epsilon$-block is also hard.

\item \textbf{Our main result: A $(1-1/e)^2$-approximation algorithm through locally-adaptive policies.} 
In Section~\ref{sec:local-adaptive} we describe our main algorithm designed for any monotone submodular function using value oracles. The algorithm finds a locally-adaptive policy whose value is guaranteed to be at least about $1-1/e$ of the value of the optimum locally-adaptive policy.  Naturally, it remains to prove that the best locally-adaptive solution is $1-1/e$ away from the true optimum. We establish this by showing that {\em any non-adaptive policy can be approximated arbitrarily well by a locally-adaptive policy}, and utilizing the bounds we have for such policies.  The idea of locally-adaptive policies is new and may be of independent interest.

\item \textbf{Adaptivity gap for locally-adaptive policies.} In section~\ref{sec:localGap} we exhibit a gap ($\approx 0.853$) between the optimal locally-adaptive policy and the optimal adaptive policy.  

\end{itemize}

\ifFULL
\subsection{Related work}

\paragraph{Adaptive Seeding.}  The Adaptive Seeding model was introduced by Seeman and Singer in~\cite{SS13}.
They use a concave relaxation to achieve a constant bound approximation for the adaptive seeding problem with influence functions in the Triggering model.
Unfortunately, their techniques do not extend to general submodular functions.
In this work we introduce new non-adaptive and adaptive techniques for this problem, and achieve a better approximation bound for the Triggering model.
Most importantly, our results hold for the entire class of submodular functions.  
In a follow-up paper~\cite{RSS14} the adaptive seeding problem is studied under \emph{knapsack constraints}. 
While the techniques used in that paper are applicable here, they give a different approximation bound than what we achieve here, which as shown in that paper is actually the appropriate bound for that case.
However, as we show, in the cardinality case we can get better approximation bounds. 

\paragraph{Submodular Functions.}  
Monotone submodular set functions maximization is an extensively studied problem. Nemhauser et al.~\cite{NWF78} show a simple greedy algorithm that achieves a $(1-1/e)$-approximation for maximizing monotone submodular functions under cardinality constraints. Feige~\cite{F98} shows that this is the best possible unless P=NP.  
A continuous version of the greedy algorithm was introduced by Calinescu et al.~\cite{CCPV07} to solve a \emph{multilinear} extension of the problem which combined with the \emph{pipage rounding}~\cite{AS04} techniques is shown to give a $(1-1/e)$-approximation under matroid constraint for a special class of submodular function. Vondrak~\cite{V08} shows that similar ideas can be used to give the same approximation for all monotone submodular functions. Checkuri, Vondrak, and Zenklusen~\cite{CVZ11} extend this framework and introduce \emph{Contention Resolution Schemes} to show how to account for multiple matroid and knapsack constraints and also extend to non monotone functions. In this paper we use the contention resolution and a bound from~\cite{CCPV07} to bound the adaptivity gap between non-adaptive and adaptive policies. 

Submodular functions have been a crucial tool in the influence maximization literature. Kempe, Kleinberg, and Tardos~\cite{KKT03} show that a class of influence models called \emph{Triggering model} are all submodular functions and thus can be approximated by a greedy algorithm as discussed above. Mossel and Roch~\cite{MR07} proved a conjecture posted in~\cite{KKT03} and show that far more general influence models can be expressed by a submodular function. These models have been studied in numerous papers both theoretically and empirically~\cite{BHMW11,borgs2012influence,C08,RLK10,LAH06,LKGFVG07,KDD11}, and thus are the main focus of this paper as well.

\paragraph{Stochastic optimization.}  The adaptive seeding model is a stochastic optimization framework (see~\cite{shapiro2009lectures} for a survey).  
There has been extensive work on stochastic optimization problems in the context of approximation algorithm varying from two-stage with recourse minimization problems (\cite{gupta2004boosted},\cite{immorlica2004costs}, \cite{ravi2004hedging}, \cite{shmoys2004stochastic,shmoys2006approximation},\cite{srinivasan2007approximation}) and of maximizing an objective under a budget constraint (\cite{dean2004approximating}, \cite{gupta2012approximation},\cite{kleinberg2000allocating}). The adaptive seeding model is different from these problems as it combines a two stage model with a maximization under budget problem. 
Moreover, the constraint structure in adaptive seeding is such that the first stage decisions determine the available actions in the second stage (not just through the shared budget constraint).
Another variant of multi-stage stochastic submodular maximization was studied by Asadpour et al.~\cite{asadpour2008stochastic} and Golovin and Krause~\cite{golovin2011adaptive}.  
Despite the similar name, the adaptive seeding model is substantially different from these models, both in motivation and in the fact that the second stage problem is strongly dependent on choices made in the first stage.
In these models the algorithm has access to the entire network but has the freedom to choose one node at a time
(or group of nodes) and observe the realized value of these nodes.
Yang et al.~\cite{YHLC13-APM} and Chen et al.~~\cite{CLLR15-AIM} also study ``multi-level" models of the influence maximization problem, the latter partially inspired by the adaptive seeding model. However, their motivation, model and benchmark are substantially different from the adaptive seeding model, and apart from some special cases, their models do not have any constant factor approximations.

\fi
\section{Non-adaptive policies}\label{sec:RR-adaptive}
A non-adaptive policy is a pair of sets ${(S,T) \subseteq X \times \mathcal{N}(X)}$ where $S$ represents the set selected in the first stage and $T \subseteq \mathcal{N}(S)$ is the set selected in the second stage.  
The natural definition for feasibility would be that the policy selects at most $k$ nodes, though it is easy to construct examples that show that such strict policies have an unbounded approximation ratio.  We therefore consider \emph{relaxed} non-adaptive policies whose guarantee is to select at most $k$ nodes \emph{in expectation}, where the expectation is over the randomization in the model, i.e. the probabilities of nodes arriving in the second stage.  
In the rest of the paper we drop the \emph{relaxed} prefix and just call such policies non-adaptive policies.
We define the value of a non-adaptive policy $(S,T)$ to be ${F(T) = \sum_{i\in [m]}p(R_i)f(T \cap R_i)}$ and its cost as $|S|+\mathcal{C}(T)$, where $\mathcal{C}(T)=\sum_{i\in T} p_i$. Finding the optimal non-adaptive policy requires solving the optimization problem: 
$$\texttt{OPT}_{NA}=\max_{S,T} \left \{ F(T) \ : \ |S| + \mathcal{C}(T)  \leq k, \ S \subseteq X, T \subseteq \mathcal{N}(S) \right \}.$$
%
%
%
The crucial difference between these policies and adaptive ones is that non-adaptive policies fix a set $T$ a priori to seeing the realization, whereas adaptive policies have the luxury of selecting a different set $T_i$ for every realization $R_i$.  
Note that these policies are not a relaxation of adaptive policies nor are they a special case, since on the one hand they fix one set $T$ but on the other hand are only restricted to the budget in expectation. Thus, it is not immediately clear they are useful for approximating adaptive policies.

We first show that the adaptivity gap (i.e. the ratio between the value of the optimal adaptive policy and that of the non-adaptive policy) is exactly $1-1/e$.

\paragraph{A tight $1-1/e$ adaptivity gap.}
Consider an instance with a single node in $X$ connected to $n =1/\delta^2$ nodes, each appearing with probability $\delta$, for some small $\delta>0$.  The function is:
$$
f(T) = \begin{cases} 1 &\mbox{if } T \neq \emptyset \\ 
0 & \mbox{otherwise}. \end{cases}
$$

For a budget of $2$, an optimal adaptive policy seeds the single node in $X$, waits for the realization of its neighbors, and seeds whichever node realizes.  The optimal non-adaptive policy here seeds the single node in $X$ and spends the rest of its budget on $1/\delta$ nodes in $\mathcal{N}(X)$, which has an expected utility of $1-(1-\delta)^{1/\delta} \approx (1-1/e)$.  The adaptivity gap is therefore at least $1-1/e$.
\ifFULL

We
\fi
complement the above example by showing that the adaptivity gap is at most $1-1/e$.  
The proof utilizes an interesting connection between the optimal adaptive policy and the \emph{concave closure} of the underlying submodular function~\cite{CCPV07}. Using this connection and properties of the concave closure from~\cite{CCPV07} we can bound the error of a \emph{fractional} non-adaptive policy that we then round to get our result. 
\ifFULL
\begin{lemma}\label{lemma:naToA}
For any budget $k$, $\texttt{OPT}_{NA}\geq  (1-1/e-2/k)\texttt{OPT}_{A}$.
\end{lemma}
\begin{proof}

Let $S$ be the set chosen by the optimal adaptive policy and let $T_i$ be the set chosen by this policy in realization $R_i$.  
For a set $T\subseteq\mathcal{N}(S)$ let $\alpha_T$ be the total probability that $T$ is chosen by the adaptive policy in the second stage. That is, $\alpha_T=\sum_{i\in\{i|T=T_i\}}p(R_i)$. 
For $i\in\mathcal{N}(S)$ let $\myvec{q}_i$ be the probability that $i$ is seeded over all realizations in the second stage. That is, $\myvec{q}_i=\sum_{\{T|i\in T\}}\alpha_T$.
Consider the function:
$$f^+(\myvec{q})=\max_{\vec{\beta}}\{\sum_{T\subseteq\mathcal{N}(S)}\beta_Tf(T)|\sum_T\beta_T=1;\beta_T\geq0;\sum_{T}\beta_T1_T=\myvec{q}\}.$$
Obviously, $\sum_T\alpha_T=1$ and $\sum_{T}\alpha_T1_T=\myvec{q}$. Thus $\vec{\alpha}$ is a valid $\vec{\beta}$ for which $f^+(\myvec{q})$ optimizes over. Since $\texttt{OPT}_{A}=\sum_T\alpha_Tf(T)$ this mean that $f^+(\myvec{q})\geq \texttt{OPT}_{A}$.

Consider instead the process in which at the second stage each element $y_i\in \mathcal{N}(S)$ is chosen independently with probability $\myvec{q}_i$ and call the value of this process $F(\myvec{q})$. By a consequence of \cite[Lemma 5]{CCPV07} we know that $F(\myvec{q})\geq (1-1/e)f^+(\myvec{q})\geq (1-1/e)\texttt{OPT}_{A}$.
Now note that $\myvec{q}_i \leq \myvec{p}_i$ as an item can't be seeded with a higher probability than the probability it realizes. 
Thus, there exists $\myvec{t}\in[0,1]^{|\mathcal{N}(S)|}$ such that for every $i\in\mathcal{N}(S)$  $\myvec{q}_i=\myvec{t}_i\myvec{p}_i$. Since in each realization only $k-|S|$ element are chosen, we know that $\myvec{t}^T\myvec{p}<k-|S|$. 
Thus, $\myvec{t}$ can be understood as a fractional solution for the second stage set of the non-adaptive policy.  We can round $\myvec{t}$ using the \emph{pipage rounding}~\cite{AS04} technique in order to get a vector with only one fractional solution and no loss in value. 
Note that we can assume without loss of generality that the fractional entry has the smallest marginal density out of all non-zero entries and that there are no fractional entries if $\myvec{t}^T\myvec{p}<k-|S|$.

First assume that $|S|< k/2$.
To get $T$ we take the items for which the rounded vector entry is $1$. From submodularity and the fact that $\myvec{t}^T\myvec{p}\geq k/2$ we get that the solution is a valid non-adaptive policy and has at most $2/k$ loss of the fractional solution.
If $|S|\geq k/2$ we can instead remove an item of $S$ (and all the entries that are connected only to it) with the least marginal value and include the fractional entry in $T$ (if it is not connected to that item). From submodularity, we can divide the value of the solution between the different items of $S$ such that at least one of them has marginal value less than $2/k$ so we get a valid solution with at most $2/k$ loss.
So we constructed a valid non-adaptive policy whose value is at least $(1-1/e-2/k)\texttt{OPT}_{A}$ and thus get our result.
\end{proof}

We next
\fi
show that a non-adaptive policy can be converted to an adaptive policy with small loss by using the contention resolution scheme~\cite{CVZ11}. %
\ifFULL
\begin{lemma}\label{lemma:nonAdaptToAdapt}
For every  $\epsilon\in(0,1/5)$ and for any non-adaptive policy \mbox{$(S,T)$} 
such that \mbox{$k-|S|>\epsilon^{-4}$}, 
there exist an adaptive policy with value 
\mbox{$\geq \left (1-2\epsilon\right)F(T)$}.
\end{lemma}

\begin{proof}
Given a solution $(S,T)$ consider an adaptive policy which seeds the same set $S$ in the first stage; in the second stage, for every realization $R$ of neighbors of $S$ the policy selects each node $j\in R\cap T$ with probability $(1-\epsilon)$ into a set $\hat{T}$ and seeds the nodes in $\hat{T}$ if $|\hat{T}|>\xi=k-|S|$ and otherwise does not seed any nodes.
We next compute the probability of any element $j$ to be seeded given that it is in $\hat{T}$.
\begin{align*}
\Pr\left [j \text{ is seeded} \ | \ j\in\hat{T}\right ]&=\Pr\left [|\hat{T}|\leq \xi \ | \ j\in\hat{T}\right ]\\
&=\Pr\left [|\hat{T}\setminus \{j\}|\leq \xi-1\right ]\\
&=1-\Pr\left [|\hat{T}\setminus \{j\}|> \xi-1\right ]\\
&\geq 1-\Pr\left[|\hat{T}\setminus \{j\}|> \left ( \frac{\xi-1}{(1-\epsilon)(\xi-\myvec{p}_j)} \right ) \E\left [|\hat{T}\setminus\{j\}| \right ]\right ]\\
&\geq 1-\exp\left(-\frac{(1-\epsilon)(\xi-\myvec{p}_j)}{3}(\frac{\xi-1}{(1-\epsilon)(\xi-\myvec{p}_j)}-1)^2\right)\\
&\geq 1-\exp\left (-\frac{\epsilon^2 \xi-2}{3}\right)\\
&\geq 1-\exp\left(-\epsilon^{-1}\right)\\
&\geq 1-\epsilon
\end{align*}
We derive the first inequality from 
\mbox{$(1-\epsilon)(\xi-\myvec{p}_j)\geq \E[|\hat{T}\setminus\{j\}|]$}.
Notice that since $\epsilon <0.5$, we have that 
\mbox{$\frac{\xi-1}{(1-\epsilon)(\xi-\myvec{p}_j)}\in(1,2)$},
and thus the second inequality follows form Chernoff bound (See appendix for the exact bound used). The third inequality follows from simple arithmetic derivations. The fourth is from substituting $\xi$ with $\epsilon^{-4}$ and using the fact that $\epsilon<1/5$.

We therefore have that the probability $j$ is seeded given that it is was realized (in $R$) is at least $(1-\epsilon)^2\geq (1-2\epsilon)$.  
We also have that the seeded set is always of size at most $\xi$. 
In addition for any element $j$ and any two realizations $R_1$ and $R_2$ such that $j\in R_1\subseteq R_2$ we have that the probability $j$ is seeded when the realization is $R_1$ is higher than if the realization is $R_2$.
These three condition define a monotone $(1-2\epsilon)$-balanced \emph{contention resolution scheme}~\cite{CVZ11} and thus using the results from~\cite{CVZ11} we get that the expected value of this process (and thus of the adaptive policy) is at least \mbox{$(1-2\epsilon)F(T)$}.
\end{proof}

Combining these results we can prove the following theorem:
\fi
%

\begin{theorem}\label{thm:NAtoA}
For every $\epsilon>0$, given an algorithm that finds a \textbf{non-adaptive} policy with value at least $\gamma \texttt{OPT}_{NA}$, there is a $ (1-1/e)\gamma - \epsilon $ approximation algorithm for the optimal \textbf{adaptive} policy. 
\end{theorem}

\ifFULL
\begin{proof}
 
Run the non-adaptive algorithm to get a non-adaptive policy $(S,T)$, with an approximation of $\gamma >0$ for $\texttt{OPT}_{NA}$.
First, assume that both $k>4/\epsilon$ and $k-|S|>(4/\epsilon)^4$. 
We can then use the same adaptive policy as in Lemma~\ref{lemma:nonAdaptToAdapt} with parameter $\epsilon/4$ and let $Adapt(S,T)$ be the value of that policy.
  
\begin{align*}
Adapt(S,T) 
&\geq \left(1-\epsilon/2\right)F(T)\\
&\geq \left(1-\epsilon/2\right)\gamma \texttt{OPT}_{NA}\\
&\geq  \left ((1-\epsilon/2) (1-1/e-\epsilon/2)\right)\gamma\texttt{OPT}_{A}\\
&\geq  \left ( (1-1/e)\gamma - \epsilon \right )\texttt{OPT}_{A}\\
\end{align*}
where the first inequality is due to Lemma \ref{lemma:nonAdaptToAdapt} and the third is due to Lemma~\ref{lemma:naToA}. 

If $k-|S|<c=(4/\epsilon)^4$ we can iteratively remove from $S$ the \mbox{$\lceil c \rceil$} elements that contribute the least to the value of the solution, as well as the elements of $T$ that are connected only to the removed elements. 
Let $S',T'$ be the result of this procedure. Notice that $k-|S'| >c$. As we removed the elements with the least value we know that $F(T')\geq(1-\frac{\lceil c \rceil}{|S|})F(T)$, and thus by the same argument as in the previous case we get that $Adapt(S',T')\geq (1-\frac{\lceil c \rceil}{|S|}) Adapt(S,T)$. It is easy to check that if $k>O(\frac{\lceil c \rceil}{\epsilon})$ we still get the desired approximation ratio.

If $k$ does no satisfy one of the conditions above (so smaller then some constant) we can find an optimal adaptive policy by a brute force search over sets of size at most $k$ in the first stage (we can approximate their value to any desired accuracy by sampling realization and finding the optimal second stage set).
\end{proof}
\fi


 


\subsection{Optimization via Non-Adaptive Policies}\label{apx-rr}
Given the blackbox reduction in Theorem~\ref{thm:NAtoA}, one can consider the problem of designing algorithms for non-adaptive policies.  
%
%
\ifFULL
We now describe the simple greedy algorithm \textsc{NonAdaptiveGreedy}
\fi
which is similar to the one sketched in the Introduction:
at each step, as long as it doesn't exceed the total budget, the algorithm adds the densest $\epsilon$-block.  For non-adaptive policies, an $\epsilon$-block is a node $x \in X$ and a subset of its neighbors whose \emph{expected} cardinality is at most $1/\epsilon$.  
\ifFULL
A formal description of the algorithms follows in Algorithm~\ref{alg:DensestAscent}. It assume a black-box access to an algorithm that finds an approximate optimal $\epsilon$-block called \textsc{FindOptimalNonAdaptiveBlock}.  

\begin{algorithm}
\caption{\textsc{NonAdaptiveGreedy}}
\label{alg:DensestAscent}
\begin{algorithmic}[1]
\INPUT $f:2^{\mathcal{N}(X)}\rightarrow \RR_+$, budget $k$.
\STATE $S \leftarrow \emptyset$, $\myvec{q}\leftarrow\overrightarrow{0}$.
\WHILE {$|S| + \sum_{j\in T}p_{j}  \leq k -\frac{3}{\epsilon}$}
\STATE $(x,B) \leftarrow \textsc{FindOptimalNonAdaptiveBlock}(S,T)$
\STATE $(S,T)=(S\cup x, T\cup B)$
\ENDWHILE
\RETURN $(S,T)$
\end{algorithmic}
\end{algorithm}
The next lemma shows that for any $\alpha>0$, a procedure which guarantees an $\alpha$-approximation for the optimal $\epsilon$-block translates to a $(1-1/e^{\alpha}-\epsilon)$-approximation guarantee for the optimal non-adaptive policy.  

\begin{lemma}\label{lem:ascent}
$\forall \epsilon>0$, assume that in every iteration \textsc{FindOptimalNonAdaptiveBlock} returns a block which is an $\alpha$-approximation to the optimal $\epsilon$-block.  
Then, when $k=\Omega(1/\epsilon^2)$
Algorithm \textsc{NonAdaptive} returns a solution $(S,T)$ such that $F(T)\geq \left (1-1/e^{\alpha}-O(\epsilon)\right)\texttt{OPT}_{NA}$.
\end{lemma}

\begin{proof}
%

Let $(S_j,T_j)$ be the solution at the beginning of iteration $j$. 
For a block $(x,B)$ let ${F_{T_j}(B) = \sum_{i\in [m]}p(R_i)(f((T_j\cup B) \cap R_i)-f(T_j \cap R_i))}$.
Let $(x_{j},B_{j})$ be the solution returned by \textsc{FindOptimalNonAdaptiveBlock} at iteration $j$ and let $(x_O,B_O)$ be the optimal $\epsilon$-block in this iteration. 
First we observe that similarly to the non-stochastic case the optimal $\epsilon$-block is a $(1-2\epsilon)$-approximation of the densest block. Consider an iteration $j$. For any block $(x,B)$ we have that if $C(B)<1/\epsilon$ than the optimal $\epsilon$-block has at least the same marginal density. Otherwise, let $B_\epsilon$ be the set of elements of $B$ of highest marginal value of cost at most $1/\epsilon$.

$$\frac{F_{T_j}(B_O)}{1+C(B_O)}\geq  \frac{F_{T_j}(B_\epsilon)}{1+C(B_\epsilon)}
\geq   \frac{\frac{1/\epsilon-1}{ C(B)}F_{T_j}(B)}{1+1/\epsilon}\geq\frac{1-\epsilon}{1+\epsilon}\frac{F_{T_j}(B)}{C(B)}\geq (1-2\epsilon)\frac{F_{T_j}(B)}{1+C(B)}$$
where the first inequality is because $B_O$ is the optimal $\epsilon$-block and $B_\epsilon$ is a candidate, the second is because $B_\epsilon$ is at least of size $1/\epsilon-1$ and from the submodularity of the function, and the other steps are just simple algebra.

We can think of the optimal non-adaptive solution as a set $O\subseteq X$ and arbitrarily partition the nodes in $\mathcal{N}(O)$ such that for each node $o_{\ell} \in O$ we associate a set of children $O_{\ell} \subseteq \mathcal{N}(o_{\ell})$.  The cost associates with each node and its children is simply $1+\mathcal{C}(O_{\ell})$. Thus, we have that:
$$ \frac{F_{T_j}(B_j)}{1+\mathcal{C}(B_j)} \geq \alpha\frac{F_{T_j}(B_O)}{1+\mathcal{C}(B_O)}\geq \alpha\left ( 1-2\epsilon \right )\max_\ell\frac{F_{T_j}(O_\ell)}{1+\mathcal{C}(O_\ell)}\geq
\alpha\left ( 1-2\epsilon \right )\frac{\sum_\ell F_{T_j}(O_\ell)}{\sum_\ell 1+\mathcal{C}(O_\ell)}\geq
\alpha\left ( 1-2\epsilon \right )\frac{F_{T_j}(O)}{k}$$
where the second inequality is from the last equation, and the last is from submodularity.

We proceed by induction to show that at any iteration we have that 

$$F(T_{j+1})\geq \left (1- \prod_{\ell=1}^j \left (1- \alpha\left ( 1-2\epsilon \right )\frac{1+\mathcal{C}(B_\ell)}{k} \right )\right )\texttt{OPT}_{NA}$$

The base case is trivial. Now assume it is true for $F(T_j)$. Then we have that
\begin{align*}
F(T_{j+1})\geq & F_{T_j}+\alpha\left ( 1-2\epsilon \right )\frac{1+\mathcal{C}(B_j)}{k}(\texttt{OPT}_{NA}-F({T_j}))\\
= & \left  (1-\alpha\left  ( 1-2\epsilon \right )\frac{1+\mathcal{C}(B_j)}{k} \right) F({T_j})+\alpha\left  ( 1-2\epsilon \right )\frac{1+\mathcal{C}(B_j)}{k} \texttt{OPT}_{NA}\\
\geq & \left  (1-\alpha\left ( 1-2\epsilon \right )\frac{1+\mathcal{C}(B_j)}{k} \right) \left (1- \prod_{\ell=1}^{j-1} \left  (1- \alpha\left  ( 1-2\epsilon \right )\frac{1+\mathcal{C}(B_\ell)}{k} \right )\right )\texttt{OPT}_{NA}\\
&+\alpha\left  ( 1-2\epsilon \right )\frac{1+\mathcal{C}(B_j)}{k}\texttt{OPT}_{NA}\\
= & \left (1- \prod_{\ell=1}^j \left  (1- \alpha\left  ( 1-2\epsilon \right )\frac{1+\mathcal{C}(B_\ell)}{k} \right )\right )\texttt{OPT}_{NA}
\end{align*}
where the first inequality is from the previous derivation and the second inequality is by the induction hypothesis.

After the last iteration the cost of the solution is greater than $k-\frac{3}{\epsilon} > \left(1- 3\epsilon \right)k$.
Therefore, 
\begin{align*}
F(T)\geq &  \left (1- \prod_{\ell=1}^t \left (1-\alpha ( 1-2\epsilon )\frac{1+\mathcal{C}(B_\ell)}{k} \right )\right )OPT_{NA}\\
\geq &\left (1- (1-  \frac{\alpha( 1-2\epsilon)\left(1- 3\epsilon \right)}{t})^{t}\right )\texttt{OPT}_{NA}\\
\geq & \left (1-\frac{1}{e^{ \alpha( 1-5\epsilon  )}}\right)\texttt{OPT}_{NA}\\
\geq &  \left (1-\frac{1}{e^{\alpha}}-O(\epsilon)\right)\texttt{OPT}_{NA}
\end{align*}
Where the second inequality is because setting all of the $C(B_j)$ to be equal minimizes the function and we know that their sum is at least $\left(1- 3\epsilon \right)k$.
\end{proof}
\fi

\ifFULL
\subsection{Finding optimal non-adaptive $\epsilon$-block}

Lemma~\ref{lem:ascent} shows that finding good approximation for non-adaptive policies reduces to computing approximations for the optimal $\epsilon$-block. In this section we show that for some special cases we can find an optimal $\epsilon$-block and thus a $(1-1/e)$ approximation for the optimal non-adaptive policy in polynomial time. Unfortunately, for general submodular functions we show it is unlikely that it can be approximated arbitrary well.  

\subsubsection{Approximating optimal $\epsilon$-block for large probabilities} 
In the first special case we consider, all the probabilities on nodes are larger than some constant.
Here, the optimization problem is easy.  
Given some constant $\epsilon>0$ the algorithm simply enumerates over all $x\in X$ and over all possible subsets of items $T\in\mathcal{N}(x)$ s.t. $C(S) \leq 1/\epsilon$.

\begin{cor}\label{thm:rnr-constant}
Let $\delta=\min_{i \in [n]} p_{i}$. 
Then for any constant $\epsilon > 0$, 
we can approximate the optimal non-adaptive policy to within $(1-1/e-\epsilon)$ 
in time $\poly(n^{1/\delta})$.
\end{cor}

In the rest of this section we turn to the more challenging task of seeding nodes with small probabilities.

\fi
\ifFULL
\subsubsection{Approximating the optimal $\epsilon$-block for MRS objective}
\fi
In case the probabilities are small enumerating over all possible solutions is computationally infeasible.  The problem of finding $\epsilon$-blocks reduces to the following fundamental problem.
\ifFULL
\begin{definition}
\textsc{Submodular-optimization-with-small-probabilities-$\delta$} (SOSP-$\delta$) problem:
We are given a monotone submodular function $f$ and probabilities of each element realizing $\myvec{p}$.
The probabilities satisfy $\max_i \myvec{p}_i \leq \delta$.
Our goal is to find a set $T$ of expected size  $k$ that maximizes the expected value of $f$. 
\end{definition}

We'll look at the fractional version of this problem in which the items of $T$ are chosen independently with probabilities $\myvec{q}$ (this can be thought of as a multi-linear relaxation~\cite{CCPV07} of SOSP). Since we are only interested in small values of $\delta$ it is easy to round fractional solutions using the \emph{pipage rounding}~\cite{AS04} technique with very small loss. Formally, we want to solve

\begin{align*}
\max_{\myvec{q}} \; \; \;& \E\left[f\left(T\right)\right]
         =\sum_{T}\left( \prod_{i\in T}\myvec{q}_{i} \prod_{i\notin T}(1-\myvec{q}_{i})\right)f\left(T\right)\\
\text{s.t. } \; \; \; &\sum_{i}\myvec{q}_{i}\leq k\\
 &\myvec{q}_i \in [0,\myvec{p}_i] \ \ \ \forall i\in [n]
\end{align*}

\fi

At a first glance, it may seem like no algorithm should be able to get an approximation better than $1-1/e$ for this problem: when $\delta=1$ the problem identifies with submodular maximization under a cardinality constraint, and due to Feige we know that no algorithm can do better than $1-1/e$ unless P=NP even for coverage functions~\cite{F98}%
\ifFULL
(even for the fractional version)%
\else
\fi
.  It seems like shrinking the constraint polytope by a factor of $\delta$ should not make a difference in the optimization.  
\ifFULL
Surprisingly, as we next show, for submodular functions in a class known as \emph{matroid rank sum} (which includes coverage functions), the above optimization problem can be solved nearly optimally.  At a high level, we show that for such functions the problem can be well approximated through a convex program, which then enables us to produce a solution whose approximation becomes optimal as $\delta$ vanishes.~\footnote{Note that this is not due to the low cost of elements which allows for example for a greedy algorithm to be nearly optimal for the knapsack problem.}

\begin{theorem}\label{lemma:MSRsmall}
Suppose that $f$ can be represented as a matroid rank sum (MRS) function. 
Then, there exists a $\left(1-\delta/2\right)$-approximation algorithm for {\sc SOSP}-$\delta$ using convex programming.
\end{theorem}

\begin{proof}
Suppose that we relax the program, so that the probability of seeding
node $i$ is $1-e^{-\myvec{q}_{i}}$ (but the cost remain the same).
We can now optimize the following program:
\begin{gather*}
\label{eq:MRS-convex}
\max_{\myvec{q}}\sum_{T}
	  \left(
		\prod_{i\in T}(1-e^{-\myvec{q}_i})
		\prod_{i\notin T}(e^{-\myvec{q}_i})
	\right)
	f(T)\\
\sum_i\myvec{q}_i\leq 1/\epsilon \\
\myvec{q}_i \in [0,\myvec{p}_i] \ \ \forall i\in [n]
\end{gather*}
Dughmi et al.~\cite{DRY11} consider essentially the same program in the context of Poisson rounding and show that this program is concave when $f$ is a MRS function. 
(They use it to achieve a $(1-1/e)$-approximation for general probabilities and do not consider the special case of small $\delta$.)
Thus, we can optimize this program (to within arbitrarily good approximation) in polynomial time.

Observe that $1-e^{-\myvec{q}_{i}}\leq \myvec{q}_{i}$. Therefore we only
decreased the probability of seeding each node, so by monotonicity of $f$,
our expected value will be at least as good as the solution of the new concave program.

We lose at most a factor of $(1-e^{-\delta})/\delta$ for any submodular function $f$ (and in particular for matroid rank sum).
Think of the process where the elements are added one by one, 
and consider the marginal contribution of each one. 
Let $\mu=\mu\left(\myvec{q}\right)$ denote the original distribution on sets\
(i.e. $\Pr_{\mu}[T] = \prod_{i \in T} \myvec{q}_i \cdot \prod_{i \notin T} (1-\myvec{q}_i)$), 
and let $\nu=\nu\left(\myvec{q}\right)$ denote the transformed distribution
(i.e. $\Pr_{\nu}[T] = \prod_{i \in T} (1-e^{-\myvec{q}_i}) \cdot \prod_{i \notin T} e^{-\myvec{q}_i}$); 
let $F_{\mu}$ and $F_{\nu}$ denote the value of the objective function under each distribution.
\begin{eqnarray*}
F_{\nu}\left(\myvec{q}\right) & = & \sum_{i}\left(1-e^{-\myvec{q}_{i}}\right)\E_{T\sim\nu}\left[f\left(\left\{ i\right\} \cup\left(T\cap\left[i-1\right]\right)\right)-f\left(T\cap\left[i-1\right]\right)\right]\\
 & \geq & \sum_{i\colon a\in A_{i}}\left(\frac{1-e^{-\delta}}{\delta}\cdot \myvec{q}_{i}\right)\E_{\color{red}T\sim\nu}\left[f\left(\left\{ i\right\} \cup\left(T\cap\left[i-1\right]\right)\right)-f\left(T\cap\left[i-1\right]\right)\right]\\
 & \geq & \sum_{i\colon a\in A_{i}}\left(\frac{1-e^{-\delta}}{\delta}\cdot \myvec{q}_{i}\right)\E_{\color{blue} T\sim\mu}\left[f\left(\left\{ i\right\} \cup\left(T\cap\left[i-1\right]\right)\right)-f\left(T\cap\left[i-1\right]\right)\right]\\
 & = & \frac{1-e^{-\delta}}{\delta}\cdot F_{\mu}\left(\myvec{q}\right)\,.
\end{eqnarray*}
The first step follows by considering the expected increment for adding
$i$, with respect to $\nu$. The second step follows by lower bounding
$e^{-\myvec{q}_{i}}$ and monotonicity. The third step follows by submodularity.
Finally, the last step follows by again considering the expected increment
for adding $i$, this time with respect to $\mu$.
\end{proof}

Finally, we obtain a tight approximation for the non-adaptive solutions to adaptive seeding instances with MRS objective and arbitrary probabilities by carefully combining the two special cases.

\begin{theorem}
\label{thm:MRS}
For any $\epsilon>0$ there is a polynomial-time algorithm that returns a $(1-1/e-\epsilon)$-approximation of the optimal non-adaptive policy for any matroid rank sum (MRS) function.
\end{theorem}

\begin{proof}
Run Algorithm \textsc{NonAdaptive} with subroutine \textsc{FindOptimalNonAdaptiveBlock} that finds $\epsilon'$-blocks implemented as follows:
Enumerate over all feasible subsets of nodes with probabilities at least $\delta$.
For each subset, let $k'$ be the remaining budget for this block. Solve the concave program for budget in $\{\epsilon'',2\epsilon'',\ldots,k'\}$. 
By Lemma \ref{lemma:MSRsmall}, 
when our enumeration reaches the optimal subset of large-probabilities elements, and we use the approximately correct additional budget (we spend at most an additional $\epsilon$ budget),
the solution of the concave program is a $\left(\left(1-\delta/2\right)-\epsilon''\right)$-approximation to the densest subset.

The concave program might return a solution that does not correspond to a set (have some $\myvec{q}_i$ that does not equal to $0$ or $\myvec{p}_i$). However, using the \emph{pipage rounding}~\cite{AS04} technique the algorithm can round it (make $\myvec{q}_i$ equals either $\myvec{p}_i$ or $0$) to have at most one undetermined item without any loss of value. If such an item remains, the algorithm compares the density of the solution that includes that item to the solution that does not include it and chooses the one with maximum marginal density. It is easy to verify that one of those solution has a higher density than the density of the fractional solution. This procedure might cause the block to cost $\delta$ more so in total $1+1/\epsilon'+\delta$.

We chose $\epsilon',\epsilon''$ and $\delta\leq \epsilon'$ (thus the total cost of a block is at most $3/\epsilon'$) such that the total loss due to Lemma \ref{lem:ascent} is $\epsilon$. 
For large values of $k$ the theorem follows by Lemma \ref{lem:ascent} and by noticing that in the analysis of that lemma we only analyze iterations where there is at least $3/\epsilon'$ budget left so this procedure returns a valid block in each such iteration. If $k$ is not large enough (smaller than some constant that depends on $\epsilon$), we can enumerate over all first stage set of size at most $k$ and complete the solution by solving a monotone submodular maximization on the second stage to get a $(1-1/e)$ approximation.
\end{proof}
\fi

Combining this with the results of the previous section we get that for MRS functions we have a $(1-1/e)^2$ approximation for the adaptive seeding problem as desired.

\ifFULL
\subsubsection{Hardness of approximating optimal $\epsilon$-block for general submodular functions}
\fi
Unfortunately, for general submodular functions the Submodular-optimization-with-small-probabilities problem cannot be approximated arbitrary well even with a constant budget (and as a special case, computing optimal densest $\epsilon$-blocks is hard).  
\ifFULL
Our algorithm heavily relies on the reduction to a concave program,
yet as pointed out by \cite{DRY11, DV11}, this program is not concave in general. While this means our current approach fails it does not mean the problem is computationally hard. 

Standard techniques for showing hardness of submodular maximization
seem to fail for this problem: Feige's construction \cite{F98} would also show
hardness for the max-cover version, but we know that this problem
is easy in this setting. The symmetry gap \cite{Vondrak13-symmetry_gap}
should give an information-theoretic lower bound (in the oracle model), 
but with a constant budget it is easy to design an exponential-time,
poly-information algorithm that achieves an arbitrarily good approximation.

We therefor look for a construction where the optimal set behaves \emph{locally}
very differently from a random set. In other words, when we look at
a random $(1/\epsilon)$-subset of the optimal set, it should be different from
what we expect from a random $(1/\epsilon)$-subset somewhere else in the graph.
Intuitively, this is very similar to the \textsc{Planted-Clique}
problem, where every subset of the clique is of course also a clique,
but the rest of the graph may be arbitrarily (constant) sparse \cite{densest_k-subgraph_AAMMW11}.
\fi

\begin{theorem}
\label{thm:planted-clique}
If the Submodular-optimization-with-small-probabilities problem with \textbf{a constant budget} $k$ can be approximated within any constant factor better than
$\left(1-e^{-k/2}\right)/\left(1-\left(\frac{k}{2}+1\right)e^{-k}\right)$,
then there is a polynomial time algorithm for the \textsc{Planted-Clique} problem that succeeds with high probability.
In particular, for $k=1.7$,
\begin{eqnarray*}
\left(1-e^{-k/2}\right)/\left(1-\left(\frac{k}{2}+1\right)e^{-k}\right)\approx0.865
\end{eqnarray*}
\end{theorem}

\ifFULL
\begin{proof}
We reduce from the Densest $l$-subgraph problem. 
In particular, Alon et al. \cite{densest_k-subgraph_AAMMW11}
proved that for any constant $\epsilon > 0$, 
given a graph $G=\left(V,E\right)$ it is \textsc{Planted-Clique} hard%
\footnote{See \cite{densest_k-subgraph_AAMMW11} for precise statement.}
to distinguish between: 
\begin{itemize} 
\item Completeness: $G$ contains a clique of size $l$; or
\item Soundness: every $l$-vertex subgraph of $G$ is $\epsilon$-sparse
(i.e. it contains at most an $\epsilon$-fraction of the ${l}\choose{2}$ edges an $l$-clique would contain.)
\end{itemize}
Given $G$, we construct a monotone
submodular function that gives $1$ for every subset that contains an edge,
and otherwise approaches $1$ exponentially with the size of the subset:
\[
f\left(T\right)=\begin{cases}
1 & \exists\left(u,v\right)\in E\mbox{ s.t. }\left\{ u,v\right\} \subseteq T\\
1-2^{-\left|T\right|} & \mbox{otherwise}
\end{cases}
\]
We set the maximal probabilities such that a budget of $k$ is sufficient
to bid exactly for the entire $l$-clique: $p_{i}=k/l$. Monotonicity is trivial.
\begin{itemize}
\item \textbf{Submodularity:} Adding $u$ to $T$ may increase $f$ by at
most $2^{-\left|T\right|}$. However, adding $u$ to any $T'\subsetneq T$
would increase $f$ by at least $2^{-\left(\left|T'\right|-1\right)}\geq2^{-\left|T\right|}$.
\end{itemize}
Now, observe that regardless of the choice of $\mathbf{q}$, the random
variable $\left|T\right|$ (the size of the realized set) behaves
approximately like a Poisson distribution with parameter $k$. More
precisely, the total variation distance between $\left|T\right|$
and $\mathbf{Pois}\left(k\right)$ is bounded $\sum \myvec{q}_{i}^{2}\leq k^{2}/l<\epsilon$
(e.g. \cite{LeCam60_poisson_binomial}).
\begin{itemize}
\item \textbf{Completeness: }Let $\mathbf{q}^{\textrm{opt}}$ be $1$ on the
$l$-clique and $0$ otherwise; then\textbf{ 
\begin{eqnarray*}
\E\left[f\left(T^{\textrm{opt}}\right)\right] & = & 1-\frac{1}{2}\Pr\left[\left|T\right|=1\right]-\Pr\left[\left|T\right|=0\right]\\
 & \geq & 1-\frac{1}{2}\Pr\left[\mathbf{Pois}\left(k\right)=1\right]-\Pr\left[\mathbf{Pois}\left(k\right)=0\right]-\epsilon\\
 & = & 1-\left(\frac{k}{2}+1\right)e^{-k}-\epsilon
\end{eqnarray*}
}
\item \textbf{Soundness: } Suppose that every $l$-subgraph of $G$ is $\epsilon$-sparse.
Observe that by submodularity of $f$, we can assume w.l.o.g. that $\mathbf{q}^{\textrm{alg}}$ defines a set, 
i.e. it is $\myvec{p}_i$ for $l$ vertices that contain at most $\epsilon {l\choose2}$ edges (and $0$ everywhere else).
Then, 
\[
\E\left[f\left(T^{\textrm{alg}}\right)\right] 
\leq 1-\sum_{i=0}^{\infty}\Pr\left[\left|T\right|=i\right]2^{-i}+ \Pr[T \text{ contains an edge}]
\]

For each of the $\epsilon {l\choose2}$ potential edges, 
the probability that both vertices belong to $T$ 
is ${|T|\choose2} / {l\choose2}$.
Taking a union bound over all of them, we have 

\[
\Pr[T \text{ contains an edge}] \leq \epsilon {|T| \choose 2} < \epsilon |T|^2
\]

Finally, since $|T|$ is distributed $\epsilon$-like $\mathbf{Pois}\left(k\right)$,

\[
\Pr[T \text{ contains an edge}] < \epsilon \cdot (k^2 + k ) + \epsilon
\]

Therefore,

\textbf{
\begin{eqnarray*}
\E\left[f\left(T^{\textrm{alg}}\right)\right]  & 
	\leq & 1-\sum_{i=0}^{\infty}\Pr\left[\mathbf{Pois}\left(k\right)=i\right]2^{-i} + O(\epsilon)\\
 & = & 1-\sum_{i=0}^{\infty}\frac{k^{i}}{i!}e^{-k}2^{-i}+O(\epsilon)\\
 & = & 1-\left(\sum_{i=0}^{\infty}\frac{\left(\frac{k}{2}\right)^{i}}{i!}e^{-k/2}\right)\cdot e^{-k/2}+O(\epsilon)\\
 & = & 1-e^{-k/2}+O(\epsilon)
\end{eqnarray*}
}
\end{itemize}
Thus, for budget $k$, it is \textsc{Planted-Clique} hard to find
any approximation which is better than $\left(1-e^{-k/2}\right)/\left(1-\left(\frac{k}{2}+1\right)e^{-k}\right)$.
For example, set $k=1.7$ to get the approximation factor of $\left(1-e^{-1.7/2}\right)/\left(1-\left(\frac{1.7}{2}+1\right)e^{-1.7}\right)<0.865$. (Recall that since we have fractional costs, the budget may be fractional as well.)
\end{proof}
\fi
This lower bound implies that the non-adaptive framework we use here cannot obtain the $(1-1/e)^2$ approximation ratio\footnote{Note that this does not exclude worse approximations guarantees via non-adaptive policies. In~\cite{RSS14} non-adaptive policies are used to obtain a (poor) constant approximation guarantee under knapsack constraints.  In this paper, we are interested in obtaining the $(1-1/e)^2$ approximation bounds.}. This obstacle motivates our use of $\epsilon$-locally adaptive policies discussed in the following section.

\section{Approximation via $\epsilon$-locally-adaptive policies}\label{sec:local-adaptive}

In this section we prove our main result:

\begin{theorem}
\label{thm:main}
For every constant $\epsilon>0$ there is an algorithm that runs in polynomial time and returns an adaptive policy which is a $((1-1/e)^2-\epsilon)$-approximation to the optimal adaptive policy with general monotone submodular functions.
\end{theorem}

\noindent\textbf{Proof outline.}
Our entire proof relies on the novel definition of a restricted class of adaptive policies which we call {\em $\epsilon$-locally-adaptive}.
Informally, we say that a policy is $\epsilon$-locally-adaptive, if it can be divided into $\epsilon$-blocks. In this context, an $\epsilon$-block is a subset of $X$ of constant size (for technical reasons these are not singletons as in previous cases), and for each realization an adaptively chosen set of constant size of its neighbors. 
We prove that a greedy algorithm that in each iteration finds the optimal $\epsilon$-block gives a $(1-1/e-\epsilon)$-approximation to the optimal {\em $\epsilon$-locally-adaptive} policy. 
This adaptive variant of $\epsilon$-blocks allows us to find the optimal subset for each realization (much in the same way as in the warmup presented in the introduction) and thus find the optimal block. 
Thus, while the non-adaptive block structure allows for greedy optimizations, the power of adaptivity within a block circumvents the hardness result of the previous section.
We then prove that the optimal {\em $\epsilon$-locally-adaptive} policy is a $(1-1/e)$-approximation to the optimal adaptive policy by using the fact that locally-adaptive policies dominate non-adaptive policies. 
In particular, we show we can convert a non-adaptive policy to an {\em $\epsilon$-locally-adaptive} policy,
with arbitrarily small loss in value, and thus our bound follows.
A natural question is then whether locally-adaptive policies are as good as adaptive policies. We answer that question negatively by presenting an example that exhibits a gap ($\approx 0.853$) between the optimal locally-adaptive policy and the optimal adaptive policy.  
We conclude this expository subsection by formally defining $\epsilon$-locally-adaptive policies.

\begin{definition}
An (adaptive) $\epsilon$-block is a set $S\subseteq X$ of size at most $1/\epsilon^2$ and for each realization $R_i$ a set $T_{i}\subseteq \mathcal{N}(S)\cap R_i$ of size at most $2/\epsilon$. The cost of a block $B$ is $\mathcal{C}(B)=|S|+\max_i(|T_i|)$.
An {\em $\epsilon$-locally-adaptive} policy is a set $\mathcal{B}$ of (not necessarily disjoint) $\epsilon$-blocks.%
\footnote{Note that the blocks are not necessarily independent - $T_i$ can depend on the entire $R_i$ and not only on $\mathcal{N}(S)\cap R_i$.}
\end{definition}

Let $T_{i,B}$ be the set seeded by block $B$ in realization $R_i$ and let $\mathcal{T}_i(\mathcal{B})=\bigcup_{B\in\mathcal{B}}T_{i,B}$. We abuse notation and generalize the value and cost functions to be applied on these policies. That is, we let the value of such a policy be $F(\mathcal{B})=\sum_{i=1}^{m}p(R_i)f(\mathcal{T}_i(\mathcal{B}))$ and its cost \mbox{$\mathcal{C}(\mathcal{B})=\sum_{B\in\mathcal{B}}\mathcal{C}(B)$}. The optimal $\epsilon$-locally-adaptive policy with budget $k$ is then: 
%
%
$$
\texttt{OPT}_{LA}^{\epsilon} =\max_{\mathcal{B}}\left\{ F(\mathcal{B})\ : \ \mathcal{C}(\mathcal{B})\leq k,\forall B\in\mathcal{B}:
|S_B|\leq 1/\epsilon^2, \forall i:
T_{i,B} \subseteq \mathcal{N}(S_B)\cap R_i, |T_{i,B}|\leq 2/\epsilon
\right\}
$$
%
\subsection{The Algorithm}
We now describe the \textsc{LocallyAdaptiveGreedy} algorithm.
We run a greedy algorithm that in each iteration adds a new $\epsilon$-block to the current solution.
The algorithm always adds a block with an optimal marginal density, 
i.e. a block which maximizes the ratio between the expected marginal contribution and cost.
A formal description of the algorithm is included below. 
The {\sc FindOptimalAdaptiveBlock} subroutine simply enumerates over all subsets of size less than $1/\epsilon^2$ of $X$ and all budgets of size at most $2/\epsilon$ and returns the pair with the highest marginal density. 

\begin{algorithm}
\caption{\textsc{LocallyAdaptiveGreedy}}
\label{alg:AdaptiveGreedy}
\begin{algorithmic}[1]
\INPUT budget $k$,$\epsilon$
	\STATE $\mathcal{B} \leftarrow \emptyset$
\WHILE {$\mathcal{C}(\mathcal{B}) <k-\frac{3}{\epsilon^2}$}
	\STATE $ \mathcal{B} \leftarrow \mathcal{B} \cup \textsc{FindOptimalAdaptiveBlock} \left(\mathcal{B},\epsilon\right)$
\ENDWHILE
\RETURN $\mathcal{B}$
\end{algorithmic}
\end{algorithm}

\paragraph{Polynomial-size representation.} As there are possibly exponential many realizations we can't hope to output a full explicit description of a locally-adaptive policy. Our algorithm instead outputs for each block its first stage set $S_i$ and a budget $k_i$ for it to optimize in the second stage as well as an order over the blocks. At every realization the policy seeds the second stage nodes by going over the blocks by order and optimizing the choices of each block given only the choices made by the previous blocks. Note that this implicitly determines the content of each block. In our algorithm we implicitly assume the order on the blocks of $\mathcal{B}$ is the order in which the algorithms adds them to $\mathcal{B}$. 
Note that we can approximate $F(\mathcal{B})$ and $F_\mathcal{B}(B)$ (the marginal value of block $B$ for policy $\mathcal{B}$) for such a policy to any desired accuracy by sampling realizations and running this process on each of them (we thus assume in the analysis that we have an oracle for their value).

\subsection{Analysis}


\begin{lemma}
\label{lem:AdaptiveGreedy}
For any $\epsilon>0$, let $\mathcal{B}$ be the solution returned by \textsc{LocallyAdaptiveGreedy} with input $k,\epsilon$. Then, \mbox{$F(\mathcal{B})\geq \left (1-1/e - O \left( (\epsilon\sqrt{k})^{-2} \right)\right)\texttt{OPT}_{LA}^{\epsilon}$}.
\end{lemma}

\begin{proof}

We first show that the marginal density of the $\epsilon$-block chosen in each iteration of the algorithm is at least the marginal density of the optimal locally-adaptive policy. Note that we always have enough budget left to add a full sized block.
Let $\mathcal{B}_j$ be the solution at the beginning of iteration $j$ and let $B_j$ be the block added at iteration $j$ with $k_j=\max_i(|T_{i,B_j}|)$. Let $\mathcal{B}_O$ denote an optimal solution with value $\texttt{OPT}_{LA}^{\epsilon}$. 
For every iteration of the algorithm we have that:
\ifFULL
\begin{align*}
\frac{F_{\mathcal{B}_j}(\mathcal{B}_O)}{k} 
 = & \frac{\E_{R_i} \left[ f_{\mathcal{T}_i(\mathcal{B}_j)}\left (\mathcal{T}_i(\mathcal{B}_O) \right ) \right] }{k} \hspace{1in} \\
  \leq & \frac{\sum_{B \in \mathcal{B}_O}\E_{R_i} \left[ f_{\mathcal{T}_i(\mathcal{B}_j)}\left (T_{i,B})\right )\right] }{\sum_{B \in \mathcal{B}_O}\mathcal{C}(B)}\\
  \leq &  \max_{B \in \mathcal{B}_O}\frac{\E_R \left[ f_{\mathcal{T}_i(\mathcal{B}_j)}\left (T_{i,B}\right ) \right] }{\mathcal{C}(B)}\\
 \leq & \frac{\E_{R_i} \left[ f_{\mathcal{T}_i(\mathcal{B}_j)}\left (T_{i,B_j}\right ) \right]}{|S_j|+k_j}\\
= &\frac{F_{\mathcal{B}_j}(B_{j})}{|S_j|+k_j}
\end{align*}
\fi
The first inequality is from the submodularity of $f$, 
and the third holds because the algorithm enumerates over all $\epsilon$-blocks as candidates in each iteration - including over the blocks of $\mathcal{B}_O$.

\ifFULL
We proceed by induction to show that at any iteration we have that
\fi

$$F(\mathcal{B}_{j+1})\geq \left (1- \prod_{l=1}^j \left (1- \frac{|S_l|+k_l}{k}\right )\right )\texttt{OPT}_{LA}^{\epsilon}$$

\ifFULL
The base case is trivial. Assume it is true for $F(\mathcal{B}_{j})$. Then we have that
\begin{align*}
F(\mathcal{B}_{j+1})\geq & F(\mathcal{B}_{j}) + \frac{|S_j|+k_j}{k} \left(\texttt{OPT}_{LA}^{\epsilon}-F({\mathcal{B}_j})\right )\\
= & \left (1-\frac{|S_j|+k_j}{k}\right )F(\mathcal{B}_{j})+\frac{|S_j|+k_j}{k}\texttt{OPT}_{LA}^{\epsilon}\\
\geq & \left(1-\frac{|S_j|+k_j}{k}\right) \left (1- \prod_{l=1}^{j-1} \left (1- \frac{|S_l|+k_l}{k} \right )\right )\texttt{OPT}_{LA}^{\epsilon}+\frac{|S_j|+k_j}{k}\texttt{OPT}_{LA}^{\epsilon}\\
= & \left (1- \prod_{l=1}^j \left (1- \frac{|S_l|+k_l}{k} \right )\right )\texttt{OPT}_{LA}^{\epsilon}
\end{align*}
where the first inequality is from the previous derivation and the second inequality is by the induction hypothesis.
\fi

When the algorithm ends the cost of the solution is greater than $k'  > \left(1- \frac{3}{k\epsilon^2} \right)k$.
Therefore, 
\ifFULL
\begin{align*}
F(\mathcal{B})
\geq &  \left (1- \prod_{j=1}^t \left (1- \frac{|S_j|+k_j}{k} \right )\right )\texttt{OPT}_{LA}^{\epsilon}\\
\geq &\left (1- (1-  \frac{k'}{kt})^{t}\right )\texttt{OPT}_{LA}^{\epsilon}\\
\geq & \left (1-\frac{1}{e^{k'/k}}\right)\texttt{OPT}_{LA}^{\epsilon}\\
\geq &  \left (1-\frac{1}{e}-O\left(\frac{1}{k\epsilon^2}\right)\right)\texttt{OPT}_{LA}^{\epsilon}
\end{align*}
where the second inequality is 
\fi
because the solution's cost is at least $k'$ and this expression is minimized when the cost is evenly spread over the iterations.
\end{proof}

\noindent\textbf{Adaptivity gap of $\epsilon$-locally-adaptive policies.}
We now show that $\epsilon$-locally-adaptive policies can arbitrarily approximate non-adaptive policies, which implies our bound.
The high level idea is to show that there are good non-adaptive policies that have a \emph{block} structure and thus can be converted to locally-adaptive policies.
We first define the notion of $\epsilon$-local for non-adaptive policies:%
\begin{definition}
A \emph{budgeted} $\epsilon$-block is a triplet $(S,k,T)$ such that $S\subseteq X$ is of size at most $1/\epsilon^2$,$\frac{1}{\epsilon}\leq k\leq\frac{2}{\epsilon}$ and $T\subseteq \mathcal{N}(S)$ satisfies $ \sum_{j\in T}p_j \leq k$.
An {\em $\epsilon$-local} non-adaptive policy is a set $\mathcal{L}$ of budgeted $\epsilon$-blocks. The cost of $\mathcal{L}$ is $\mathcal{C}(\mathcal{L})=\sum_{B\in\mathcal{L}}|S_B|+k_B$. 
\end{definition}

We now prove that a non-adaptive policy can be converted into an $\epsilon$-local non-adaptive policy. Let $\mathcal{T}(\mathcal{L})$ be the union of all of $\mathcal{L}$'s blocks second stage sets.

\begin{lemma}
\label{lem:S_eps-q_eps}
For any non-adaptive policy $(S,T)$ in which $|S|+\mathcal{C}(T)>3/\epsilon^3$, there exists an $\epsilon$-local non-adaptive policy $\mathcal{L}$ of the same cost such that $F(\mathcal{T}(\mathcal{L}))\geq (1-3\epsilon)F(T)$.
\end{lemma}

\begin{proof}

Let $k=|S|+\mathcal{C}(T)$. Fix some order $(x_1\ldots x_{|S|})$ on the elements of $S$. 
We say that a second-stage node \mbox{$y\in \mathcal{N}(S)\cap T$} {\em belongs} to $x_j$ 
if $x_j$ is the smallest-indexed node in $S$ with an edge to $y$.
For a set $Q\subseteq S$, let $\mathcal{R}(Q)$ be the nodes that belong to nodes in $Q$.
Iteratively construct the blocks of $\mathcal{L}$ by adding nodes from $S$ into sets $S_i$ by the order $(x_1\ldots x_{|S|})$, 
until either
\begin{itemize}
\item $2/\epsilon\geq \mathcal{C}(\mathcal{R}(S_i))>{1}/{\epsilon}$ - set $T_i=\mathcal{R}(S_i) $ and $k_i= \mathcal{C}(\mathcal{R}(S_i))$.
\item $\mathcal{C}(\mathcal{R}(S_i))>{2}/{\epsilon}$ - 
let $\Phi$ be the maximal set of nodes of $\mathcal{R}(x)$, where $x$ is the last node added to $S_i$, that can be added to $\mathcal{R}(S_i\setminus x)$ such that the cost remains under $2/\epsilon$ (there is at least one such node since a node's cost is at most $1$); set $T_i=\mathcal{R}(S_i\setminus x) \cup \Phi$ and $k_i= \mathcal{C}(T_i)$;
start the next set $S_{i+1}$ again with $x$ but ignore the nodes in $\Phi$. 
\item $|S_i|={1}/{\epsilon^2}$ or we are done - set  $T_i=\mathcal{R}(S_i)$ and $k_i={1}/{\epsilon}$. 
\end{itemize}

The first condition incurs no extra cost. 
The second condition applies at most $\lceil\epsilon\mathcal{C}(T)\rceil$ times and incurs an extra cost of $1$ for the duplicated node, and in total at most $\epsilon\mathcal{C}(T)+1$.
The third condition applies at most $\lceil\epsilon^2|S|\rceil$ times and it incurs an extra cost of at most $1/\epsilon$, and in total at most $\epsilon|S|+{1}/{\epsilon}$. 
The total additional cost incurred is therefore at most $\epsilon k+1+1/\epsilon\leq 2\epsilon k$.



Iteratively remove from $\mathcal{L}$ the blocks whose marginal density is the lowest until $\mathcal{C}(\mathcal{L})\leq k$. Their total cost at most $2\epsilon k+(1/\epsilon^2+2/\epsilon)\leq 3\epsilon k$.
Thus, by submodularity, $F(\mathcal{T}(\mathcal{L}))\geq (1-3\epsilon)F(T)$. 
\end{proof}

Our last step is to show that given an $\epsilon$-local non-adaptive policy, we can construct an $\epsilon$-locally-adaptive policy with almost the same value. 
\ifFULL
The following lemma follows from Lemma~\ref{lemma:nonAdaptToAdapt} 
\fi
by analyzing each block of the policy independently. 
\ifFULL
The $\epsilon^4$ is needed to match the condition on the budget left for the second stage in that Lemma.
\fi

\begin{lemma}\label{lemma:nonAdaptToLocalAdapt}
For any $\epsilon$ such that $\epsilon \in (0,1/5)$, and for any $\epsilon^4$-local non-adaptive policy $\mathcal{L}$ 
there exists an $\epsilon^4$-locally-adaptive policy $\mathcal{B}$ with value 
\mbox{$\geq \left (1-2\epsilon\right)F(\mathcal{T}(\mathcal{L}))$}.
\end{lemma}

Putting it all together, we have that for every constant $\epsilon >0$ there exist $\epsilon_1,\epsilon_2,\epsilon_3$ such that:
\begin{gather*}
F(\mathcal{B}) \geq \left (1-{1}/{e}-\epsilon_1 \right ) \texttt{OPT}_{LA}^{\epsilon_2}
\geq \left (1-{1}/{e}-\epsilon_3 \right ) \texttt{OPT}_{NA} \geq \left ((1-{1}/{e})^2-\epsilon \right ) \texttt{OPT}_A
\end{gather*}
where the first inequality follows from Lemma \ref{lem:AdaptiveGreedy}, the second from Lemma \ref{lemma:nonAdaptToLocalAdapt} and the third from the results of the previous section.
This completes our proof of Theorem \ref{thm:main}.
\qed

\subsection{Separation between $\epsilon$-locally-adaptive and adaptive policies}\label{sec:localGap}
We complement this section with an example that exhibits a gap between the values of the optimal $\epsilon$-local-adaptive policy and the optimal adaptive policy. 
It remains open whether the optimal gap is $(1-1/e)$ as our upper bound suggests.
Any better upper bound on the gap will immediately imply a better approximation bound for Algorithm \textsc{LocallyAdaptiveGreedy}.
\begin{lemma}
There are instances where every $\epsilon$-locally-adaptive policy achieves at most $\approx 0.853$ of the value of the optimal adaptive policy.
\end{lemma}

\begin{proof}
We construct an instance where it is advantageous to move large amounts of the budget after seeing the realizations.
Since this is only possible in a fully adaptive (i.e. not $\epsilon$-locally-adaptive) policy, we obtain a separation between the two classes of solutions.

Let $m$ be a large parameter, and consider an instance where the first stage has $|X|=m$ nodes.
Each node $x \in X$ is connected to one {\em special} node $y_x \in Y$, and it also has $m^2$ {\em regular} neighbors $z_x^i \in Z_x \subset Z$.
The special nodes in $Y$ realize with probability $1/m$, while the regular nodes in $Z$ realize with probability $1$.
Let the budget be $k = m^2+m+1$.
Finally, consider the 
\ifFULL
function:
\fi
\[
f(T) = 1- \prod_{x\in X}\left( 1- \frac{1}{2}|T \cap {y_x}| - \frac{1}{2m^2}|T \cap Z_x| \right).
\]
\ifFULL
To see that $f$ is indeed submodular, consider the following observation by Dughmi and Vondrak \cite{DV11}:
For any monotone submodular functions $g_1$ and $g_2$ with values in $[0,1]$, $g(S) := 1- (1-g_1(S))(1-g_2(S))$ is also monotone submodular. 
Applying this lemma recursively, we see that for any number of monotone submodular functions $g_1 \dots g_t$ with values in $[0,1]$, 
$g(S) := 1- \prod_i(1-g_i(S))$ is also monotone submodular. Finally, $f$ is submodular since it can be written in this form for coverage functions $f_x(T) = \frac{1}{2}|T \cap {y_x}| - \frac{1}{2m^2}|T \cap Z_x|$.
\fi
\begin{description}
\item[Optimal adaptive:] 
The optimal adaptive solution seeds all the first-stage nodes. 
With probability $1-1/e$, one of the special nodes $y_x$ realizes, 
in which case the optimal policy seeds $y_x$, 
as well as all the nodes $z_x \in Z_x$. In this case, $f(T) = 1$.
With probability $1/e$, none of the special nodes realize. 
In this case, the adaptive policy picks $x$ arbitrarily, and seeds all the nodes $z_x \in Z_x$.
In this case $f(T) = 1/2$.
In expectation, $\E[f(T)] = 1-1/(2e)$.

\item[$\epsilon$-locally adaptive:]
Given a realization, we say that a block $S_i$ is {\em special} 
if it contains a node with a realized special neighbor, and {\em regular} otherwise.
Observe that in any locally adaptive policy, 
the probability that a block $S_i$ is special is less than $1/(m\epsilon^2)$ by union bound.
Therefore, the expected total budget $k_i$ 
allocated to special blocks $S_i$ is at most $O(m/\epsilon^2)$, 
i.e. negligible in comparison to $|Z_x| = m^2$.
Therefore any locally adaptive policy must spend all but a negligible amount of its budget on regular subsets. 
(It may still seed all the realized special nodes, but with only a negligible fraction of their neighbors).
Let $\xi$ denote the number of realized special nodes. 
Conditioning on $\xi$, we have,
\begin{gather*}
\E\left[f\left(T\right) \mid \xi \right] 
 \leq 1- 2^{-\xi}\cdot \prod_{x\in X}\left( 1- \frac{1}{2m^2}|T \cap Z_x| \right) + o(1).
\end{gather*}
The optimal way to spend the budget on regular blocks is to pick some $x$ arbitrarily, 
and seed (almost) all the nodes $z_x \in Z_x$, 
leaving sufficient budget to seed any realized special nodes.
In particular, for any feasible $T$,
\begin{gather*}
\prod_{x\in X}\left( 1- \frac{1}{2m^2}|T \cap Z_x| \right)  
\geq  1- \sum_{x\in X} \frac{1}{2m^2}|T \cap Z_x| 
\geq 1/2.
\end{gather*} 
Finally, notice that the distribution of $\xi$ is approximately ${\bf Pois}(1)$.  The expected value of the locally adaptive policy is therefore:
\ifFULL
\begin{eqnarray*}
\E\left[f\left(T\right)\right] 
 & \approx & 1-\frac{1}{2}\cdot\sum_{i=0}^{\infty}\Pr\left[\mathbf{Pois}\left(1\right)=i\right]2^{-i}\\
 & = & 1-\frac{1}{2}\cdot\sum_{i=0}^{\infty}\frac{1}{i!}e^{-1}2^{-i}\\
 & = & 1-\frac{1}{2}\left(\sum_{i=0}^{\infty}\frac{\left(\frac{1}{2}\right)^{i}}{i!}e^{-1/2}\right)\cdot e^{-1/2}\\
&= & 1-\frac{1}{2}e^{-1/2}
\end{eqnarray*}
\fi
\end{description}
Thus the ratio between the values of the policies is
$
\frac{1-1/(2e)}{1-\frac{1}{2}e^{-1/2}} \approx 0.853.
$
\end{proof}



\clearpage
\bibliographystyle{abbrv}
\bibliography{locally_adaptive}
\clearpage
\appendix
\section{Chernoff bound}
We use the following bound in the paper (See for example~\cite[Theorem 4.4]{MU05}).
\begin{theorem}[Chernoff bound]
Suppose $X_1,\ldots, X_n$ are independent random variables taking values in $[0,1]$ and let $X$ denote their sum and $\mu$ be the expected value of $X$. Then, for $0<\delta<1$ we have that
$$Pr[X>(1+\delta)\mu]\leq \exp(-\frac{\delta^2\mu}{3})$$
\end{theorem}

\end{document}